\documentclass[11pt]{article}
\usepackage{latexsym, amssymb, amsmath, graphicx, amsthm, url}
\usepackage[boxed,noline,linesnumbered,noend]{algorithm2e}
\usepackage[top=1in,right=1in,left=1in,bottom=1in]{geometry}
\usepackage{cite}
\usepackage{enumitem}

\newtheorem{lem}{Lemma}
\newtheorem{thm}{Theorem}

\newtheorem{defn}{Definition}
\newtheorem{prob}{Problem}

\newcommand{\one}{{\bf 1}}
\newcommand{\tbar}[1]{\mbox{%
  \sbox0{#1}\sbox2{\~{}}%
  \ooalign{\hidewidth\raise\dimexpr\ht0-\ht2+.3ex\box2 \hidewidth\cr#1\cr}}}

\title{Universal Streaming}

\author{Vladimir Braverman
\thanks{
Department of Computer Science, Johns Hopkins University.
Email:
{\tt vova@cs.jhu.edu}.  Research supported in part by
DARPA grant N660001-1-2-4014. Its contents are solely the responsibility of the author and do not
represent the official view of DARPA or the Department of Defense.}
\and
Rafail Ostrovsky\thanks{
Department of Computer Science and Department of Mathematics, University of California, Los Angeles.
Email: {\tt rafail@cs.ucla.edu}.  Research supported in part by NSF grants CNS-0830803; CCF-0916574; IIS-1065276; 
CCF-1016540; CNS-1118126; CNS-1136174; US-Israel BSF grant 2008411, OKAWA Foundation Research Award, IBM 
Faculty Research Award, Xerox Faculty Research Award, B. John Garrick Foundation Award, Teradata 
Research Award, and Lockheed-Martin Corporation Research Award. This material is also based upon 
work supported by the Defense Advanced Research Projects Agency through the U.S. Office of Naval 
Research under Contract N00014-11-1-0392. The views expressed are those of the author and do 
not reflect the official policy or position of the Department of Defense or the U.S. Government.}
\and
Alan Roytman\thanks{
Department of Computer Science, University of California, Los Angeles.
Email: {\tt alanr@cs.ucla.edu}.}
}
\date{}

\begin{document}

\maketitle
\thispagestyle{empty}
\begin{abstract}
Given a stream of data, a typical approach in streaming algorithms is to design a sophisticated algorithm
with small memory that computes a specific statistic over the streaming data.  Usually, if one wants
to compute a different statistic {\em after} the stream is gone, it is impossible. But
what if we want to compute a different statistic after the fact?  In this paper, we consider the
following fascinating possibility: can we collect some small amount of {\em specific} data
during the stream that is ``universal,'' i.e., where we do not know anything about the statistics we
will want to {\em later} compute, other than the guarantee that had we known the statistic ahead of time,
it would have been possible to do so with small memory?  In other words, is it possible to collect
some data in small space during the stream, such that any other statistic that can be computed with comparable
space can be computed after the fact?  This is indeed what we introduce (and show) in this paper
with matching upper and lower bounds: we show that it is possible to collect {\em universal}
statistics of polylogarithmic size, and prove that these universal statistics allow us {\em after the fact}
to compute all other statistics that are computable with similar amounts of memory. We show that this is
indeed possible, both for the standard unbounded streaming model and the sliding window streaming model.
\end{abstract}

\section{Introduction}
With the vast amount of data being generated today, algorithms for data streams continue
to play an important role for many practical applications.
As the amount of data being generated
continues to grow at a staggering rate, streaming algorithms are increasingly becoming more important as a practical tool
to analyze and make sense of all the information.
Data streams have received a lot of attention
with good reason, as evidenced by the wide array of applications discussed in~\cite{A07,M05}.
Applications for streaming algorithms
which operate over input that arrives on the fly and use a small
amount of memory
are numerous, ranging from monitoring packets flowing across a network to analyzing
patterns in DNA sequences.  In practice, such applications generate vast amounts of data in a very short
period of time, so it is infeasible to store all of this information.
This presents a pressing
question: when is it possible to avoid storing all the information while still providing approximate
solutions with good theoretical guarantees?

Typically, algorithms are developed for data streams in the unbounded model, where some
statistic is maintained over the entire history of the stream.  For certain applications, it is useful
to only compute such statistics over recent data.  For instance, we may wish to analyze
stock market transactions in a particular timeframe~\cite{OMMGM02} or monitor packets
transmitted over a network in the last hour to identify suspicious activity~\cite{VSGB05}.
This framework is known as the sliding window model, where we maintain statistics
over the current window of size at most $N$, which slides as time progresses.  In the sequence-based model,
exactly one element arrives and expires from the window per time step.
In the timestamp-based model, any number of elements may arrive or expire.
Clearly, the timestamp-based model is more general.

In a landmark paper that influenced the field as a whole, the work of Alon, Matias and Szegedy~\cite{AMS96}
studied the following fundamental framework.
For a universe $U = \{1,\ldots,n\}$ and an input stream (i.e., a sequence of integers drawn from $U$),
let $M = (m_1,\ldots,m_n)$ be the vector where each entry $m_i$ denotes the frequency
with which element $i \in U$ appears in the stream.  At any point in time, the paper of~\cite{AMS96}
showed how to approximate various frequency moments in sublinear space.  Informally, for the $k^{th}$ frequency
moment $F_k = \sum_{i\in U} m_i^k$, it was shown that $F_0,F_1,$ and $F_2$ can be approximated in polylogarithmic space,
while for $k > 2$, an upper bound of $O^*(n^{1-{1/k}})$ was shown (the notation $O^*(f(n))$ hides polylogarithmic factors).
Moreover, a lower bound of $\Omega(n^{1-{5/k}})$ was shown for every $k \geq 6$.  As discussed in~\cite{AMS96},
such frequency functions are tremendously important in practice and have many applications in databases, as
they indicate the degree to which the data is skewed.
The fundamental work of Indyk and Woodruff~\cite{IW05} showed how to compute $F_k$ for $k>2$ in
space $O^*(n^{1-2/k})$, which was the first optimal result for such frequency moments.
Their technique reduced the problem of computing $F_k$ to computing heavy hitters, and indeed
our construction builds on their methods.
Their result was later improved up to polylogarithmic
factors by Bhuvanagiri, Ganguly, Kesh and Saha~\cite{BGKS06}.
Recently, Li, Nguy\tbar{\^e}n, and Woodruff~\cite{LNW14} showed that any one-pass streaming
algorithm that approximates an arbitrary function in the turnstile model can be implemented
via linear sketches.  Our work is related, since our algorithms are based on linear sketches
of~\cite{AMS96}.

Such works have opened a line of research
that is still extremely relevant today.  In particular,
what other types of frequency-based functions
admit efficient solutions in the streaming setting, and which functions are inherently difficult to approximate?
In our paper, we strive to answer this question for frequency-based, monotonically increasing
functions in the sliding window model.  We make progress on two significant, open problems
outlined in~\cite{sub30} by Nelson and~\cite{sub20} by Sohler.  Specifically, we are the first
to formalize the notion of universality for streaming over sliding windows.  Our main result is
the construction of a universal algorithm in the timestamp-based sliding window model for a broad class of
functions.  That is, we define a class of functions and design a single
streaming algorithm that produces a data structure with the following guarantee.  When querying
the data structure with a function $G$ taken from the class, our algorithm approximates $\sum_{i=1}^nG(m_i)$
without knowing $G$ in advance (here, $m_i$ denotes the frequency that element $i$ appears in the window).
Our data structure only collects statistics based on
the input stream itself without knowing $G$.  The algorithm uses polylogarithmic memory in the
universe size $n$ and the window size $N$, and obtains a $(1 \pm \epsilon)$-approximation.
This is precisely the notion of universality that we develop in our paper, and
it is an important step forward towards resolving the problem in~\cite{sub30}.

Along the way, we design a zero-one law for a broader class of
monotonically increasing functions $G$ which are zero at the origin that specifies when
$\sum_{i=1}^n G(m_i)$ can be approximated with high probability in one pass, using polylogarithmic memory.
If $G$ satisfies the conditions specified by the test, then given the function $G$ we construct an
explicit, general algorithm that is able to approximate the summation to within a $(1 \pm \epsilon)$-factor
using polylogarithmic memory.  If the function $G$ does not pass the test, then we provide a lower
bound which proves it is impossible to do so.
This result generalizes the work of~\cite{BO10} to the sliding window setting,
and makes important progress towards understanding the question posed in~\cite{sub20}.

\subsection{Contributions and Techniques}
Our contributions in this paper make progress on and give insight into two important problems:
\begin{enumerate}[itemsep=0mm]
\item We are the first to formally define the notion of universality in the streaming setting.  We define
a large class of functions $\mathcal{U}$ such that, for the entire class, we design a single, universal algorithm
for data streams in the sliding window model which maintains a data structure with the following guarantee.
When the data structure is queried with any function $G \in \mathcal{U}$, it outputs a
$(1 \pm \epsilon)$-approximation of $\sum_{i=1}^nG(m_i)$ without knowing $G$ in advance (note
that the choice of $G$ can change).  Our algorithm uses polylogarithmic memory (in $n$ and $N$), makes one
pass over the stream, and succeeds with high probability.

\item We give a complete, algebraic characterization (i.e., a zero-one law) for the class of tractable functions over
sliding windows.  We define a broader set of functions $\mathcal{T}$
such that, for any non-decreasing function $G$ where $G(0) = 0$, if $G \in \mathcal{T}$, then we have an algorithm that gives
a $(1 \pm \epsilon)$-approximation to $\sum_{i=1}^n G(m_i)$, uses polylogarithmic memory, makes one pass
over the stream, and succeeds with high probability.  Moreover, if $G \not\in \mathcal{T}$,
we give a lower bound which shows that super-polylogarithmic memory is necessary in order to approximate
$\sum_{i=1}^nG(m_i)$ with high probability.  This generalizes the result of~\cite{BO10}
to the sliding window setting.
\end{enumerate}

Our algorithms work in the timestamp-based sliding window model and maintain the summation approximately for every window.
The value $\epsilon$ can depend on $n$ and $N$, so that
the approximation improves as either parameter increases. Our construction
is very general, applying to many functions using the same techniques.
In particular, streaming algorithms tend to depend specifically on the function
to be approximated.  For instance, consider the specific algorithms for $F_2$~\cite{AMS96, I06},
$F_0$~\cite{FM85, CDIM03, AMS96} and $L_p, 0\le p<2$~\cite{I06, L09, KNW10Soda} (where the $L_p$
norm is the $p^{th}$ root of $F_p$).  The problems we study have been open
for several years, and our construction and proofs are
non-trivial.  Surprisingly, despite us using existing techniques, their solutions have remained
elusive.  The techniques we use allow us to avoid the strong pseudorandom
generator machinery developed by Nisan~\cite{N90}.  In fact, our construction only assumes
$4$-wise independence, which might be useful in practice.

For our main result, item $1$, it is useful to understand our techniques for solving
item $2$.  Designing the correct characterization for ``tractable" functions is in itself a challenging task.
Indeed, one may think that the predicate from~\cite{BO10} is sufficient for designing an algorithm
in the sliding window model.  Unfortunately, this idea is difficult to carry through, and with good reason:
it turns out to be false!  Part of the novelty and difficulty of our techniques is the identification of an extra smoothing
assumption about the class of tractable functions over sliding windows.  If a function
does not satisfy our smoothing assumption, we show a super-polylogarithmic lower bound,
inspired by the proof of~\cite{DGIM02}.
We draw on the techniques of~\cite{BO10,BO13,IW05} for our positive result by first finding
heavy elements according to the function $G$, and then reducing the sum problem to the heavy elements
problem.  Our work sheds light on the question posed in~\cite{sub20},
by exhibiting a strict separation result between the unbounded and sliding window models.

To obtain our main result, we observe that one can remove the assumption from our initial constructions
that $G$ is given up front (so that all applications of $G$ happen at the end of the window).  However,
some technical issues arise, as our construction relies on some parameters of $G$ that stem from our
zero-one law.  To address these issues, we parameterize our class of functions $\mathcal{U}$ by a constant,
allowing us to build a single algorithm to handle the entire parameterized class.

\subsection{Related Work}
The paper of Braverman and Ostrovsky~\cite{BO10} is the most closely related to our paper.
We extend their result from the unbounded model to the timestamp-based sliding window model
(by formalizing a new characterization of tractable functions) and by designing a universal
algorithm for a large class of functions.  Our results build on~\cite{BO10,BO13,IW05}.

Approximating frequency moments and $L_p$ norms is
well studied in the literature, as it has many applications,
and there are indeed a vast number of papers on the subject.  Compared to such works,
we make minimal assumptions and our results are extremely broad, as we design general algorithms that can not
only handle frequency moments, but other functions as well.
Flajolet and Martin~\cite{FM85} gave an algorithm to approximate $F_0$
(i.e., counting distinct elements),
and Alon, Matias, and Szegedy~\cite{AMS96} showed how to approximate $F_k$ for
$0 \leq k \leq 2$ using polylogarithmic memory, while for $k > 2$ they showed how to approximate
$F_k$ using $O^*(n^{1-1/k})$ memory.  They also showed an $\Omega(n^{1-5/k})$ lower bound for
$k \geq 6$.  Indyk~\cite{I06} used stable distributions to approximate $L_p$ norms for
$p \in (0, 2]$.  Indyk and Woodruff~\cite{IW05} gave the first optimal algorithm
for $F_k$ ($k > 2$), where an $O^*(n^{1-2/k})$ upper bound was developed.  In a
followup work, Bhuvanagiri, Ganguly, Kesh, and Saha~\cite{BGKS06} improved the space by
polylogarithmic factors.  For lower bounds, Bar-Yossef, Jayram, Kumar, and
Sivakumar~\cite{BJKS02} gave an $\Omega(n^{1 - (2+\epsilon)/k})$ lower bound, which was
later improved to $\Omega(n^{1-2/k})$ by Chakrabarti, Khot, and Sun~\cite{CKS03} for any
streaming algorithm that makes one pass over the stream.  The literature is vast, and other
results for such functions
include~\cite{IW03,W04,BJKST02,CK04,CDIM03,FKSV99,G04,GC07,L09,KNW10Soda,KNW10Pods}.
There is also literature on finding frequent items, and this is indeed a problem we must solve
to achieve our results (although we must do so for a broad class of functions).
Examples of works which find frequent items include~\cite{CCF02,CH08,CM05ELS}.
Moreover, there has been a line of work in the literature
on estimating entropy and entropy norms, including~\cite{BG06,CBM06,CCM07,GMV06,HNO08,LSOXZ06}.

There is also a vast literature in streaming for sliding
windows.  In their foundational paper, Datar, Gionis, Indyk, and Motwani~\cite{DGIM02}
gave a general technique called exponential histograms that allows many fundamental
statistics to be computed in optimal space, including count, sum
of positive integers, average, and the $L_p$ norm for $p \in [1,2]$.
Gibbons and Tirthapura~\cite{GT02} made improvements for the sum and count problem
with algorithms that are optimal in space and time.  Braverman and Ostrovsky~\cite{BO07}
gave a general framework for a large class of smooth functions, which include the $L_p$ norm
for $p > 0$.  Our work complements their results, as
the functions they studied need not be frequency based.  Other problems include frequent itemsets~\cite{CWYM04}, frequency counts and
quantiles~\cite{AM04,LT06PODS}, rarity and similarity~\cite{DM02}, variance and $k$-median~\cite{BDMO03},
diameter in multidimensional space and other geometric problems~\cite{FKZ05,CS04,AHV05}, and uniform
random sampling~\cite{BDM02}.  Many works have studied frequency estimation and frequent item identification,
including~\cite{GDDLM03,JQSYZ03,CM05TRANS}, along with $L_1$-frequent elements,
including~\cite{HT08,ZG08,BAE07,HLT10,NL05}.  The recent work of~\cite{BGO13} gave an efficient
algorithm for computing $L_2$-frequent elements over sliding windows.  Many of our constructions rely
on computing frequent elements, but we must do so under a broad class of functions.

\subsection{Roadmap}
In Section~\ref{sec:prob}, we describe notation used throughout this paper, give some definitions, and formalize the main
problems we study.  In Section~\ref{sec:lb}, we give a lower bound for functions that are not tractable (i.e., we show the ``zero"
part of our zero-one law).  In Section~\ref{sec:tract}, we give an algorithm for any tractable function (i.e., we show the ``one" part
of our zero-one law).  Finally, in Section~\ref{sec:universality}, we show the main result of this paper by giving a universal
streaming algorithm.

\section{Notation and Problem Definition}
\label{sec:prob}
We have a universe of $n$ elements $[n]$, where $[n] = \{1,\ldots,n\}$, and an integer $N$.
A stream $D(n,N)$ is a (possibly infinite) sequence of integers $a_1,a_2,\ldots$, each from
the universe $[n]$, where $N$ is an upper bound on the size of the sliding window.  Specifically,
at each time step, there is a current window $W$ that contains \emph{active} elements, where $|W| \leq N$.
The window $W$ contains the most recent elements of the stream,
and elements which no longer belong in the window are \emph{expired}.  We use the
timestamp-based model for sliding windows (i.e., any number of elements from the stream
may enter or leave the window at each time step).  We denote the frequency vector
by $M(W)$, where $M(W) = (m_1,\ldots,m_n$) and each $m_i$ is the frequency of element $i \in [n]$
in window $W$ (i.e., $m_i = |\{j \ | \ a_j = i \wedge j \textrm{ is active}\}|$).  For the
window $W$, the $k^{th}$ frequency moment $F_k(M(W)) = \sum_{i=1}^n m_i^k$.
For a vector $V = (v_1,\ldots,v_n)$, we let $|V|$ be the $L_1$-norm of $V$, namely
$|V| = \sum_i |v_i|$.  For a vector $V = (v_1,\ldots,v_n)$ and a function $f$, we define the
$f$-Vector as $f(V) = (f(v_1),\ldots,f(v_n))$.

We say that $x$ is a $(1 \pm \epsilon)$-approximation of $y$ if $(1 - \epsilon)y \leq x \leq (1 + \epsilon)y$.  We define $O^*(f(n,N)) =~
O(\log^{O(1)}(nN) f(n,N))$.  We say a probability $p$ is negligible if $p = O^*\left(\frac{1}{nN}\right)$.  Consider the following problem:

\begin{prob}[$G$-Sum]
Let $G : \mathbb{R} \rightarrow \mathbb{R}$ be an arbitrary function which is non-decreasing such that $G(0) = 0$.  For any stream
$D(n,N)$, any $k$, and any $\epsilon = \Omega\left(\frac{1}{\log^k(nN)}\right)$, output a $(1 \pm \epsilon)$-approximation of $\sum_{i=1}^n G(m_i)$ (where $m_i$ is the multiplicity
of element $i \in [n]$) for the current window $W$.
\end{prob}

We first give some definitions which will be useful throughout the paper and help us define our notion
of tractability, beginning with the local jump:

\begin{defn}[Local Jump]
For $\epsilon > 0$ and $x \in \mathbb{N}$, we define the local jump as
$$\pi_\epsilon(x) = \min{\{x, \min{\{z \in \mathbb{N} \ | \ (G(x + z) > (1+\epsilon)G(x)) \vee (G(x - z) < (1 - \epsilon)G(x)) \}} \}} .$$
\end{defn}
\noindent That is, $\pi_\epsilon(x)$ is essentially the minimum amount needed to cause $G$ to jump by a
$(1 \pm \epsilon)$-factor by shifting either to the left or to the right of $x$.

\begin{defn}[Heavy Element]
For a vector $M(W) = (m_1,\ldots,m_n)$, a function $f$, and a parameter $d > 0$, we say that element $i$ is $(f,d)$-heavy with respect to $M(W)$ if
$f(m_i) > d\sum_{j \neq i}f(m_j)$.
\end{defn}

\begin{defn}[Sampled Substream]
Let $D(n,N)$ be a stream and $H:[n] \rightarrow \{0,1\}$ be a function.  We denote by $D_H$ the sampled substream of $D$ consisting of
all elements that are mapped to 1 by the function $H$.  More formally, $D_H = D \cap H^{-1}(1)$.
\end{defn}

\begin{defn}[Residual Second Moment]
If there is an $(F_2,1)$-heavy element $m_i$ with respect to $M(W)$, we define the residual second moment as
$F_2^{res}(M(W)) = F_2(M(W)) - m_i^2 = \sum_{j \neq i}m_j^2$.
\end{defn}

We are now ready to define our zero-one law.
\begin{defn}[Tractability]\label{defn:tractability}
We say a function $G$ is tractable if $G(1) > 0$ and:

\vspace{1mm}
$\forall k \ \exists N_0,t \ \forall x,y \in \mathbb{N^+} \ \forall R \in \mathbb{R^+} \ \forall \epsilon:$
\vspace{-4mm}
\begin{equation}\label{eqn:tract1}
\left(R > N_0, \frac{G(x)}{G(y)} = R, \epsilon > \frac{1}{log^k(Rx)} \right) \Rightarrow \left( \left(\frac{\pi_\epsilon(x)}{y}\right)^2 \geq \frac{R}{log^t(Rx)} \right) \hspace{3mm} and
\end{equation}
\vspace{-2mm}
\begin{equation}\label{eqn:tract2}
\hspace{-2.3755in}\forall k \ \exists r,N_1 \ \forall x \geq N_1 \ \forall \epsilon : \epsilon > \frac{1}{log^k(x)} \Rightarrow \pi_\epsilon(x) \geq \frac{x}{\log^r(x)}. \\
\end{equation}

\end{defn}
We let $\mathsf{Tractable}$ be the set of functions which satisfy the above predicate.
Based on this definition, we formalize the notion of tractability for our universal setting.  It is similar to the definition of tractability,
except we need to upper bound some parameters by a constant.

\begin{defn}[Universal Tractability]\label{defn:utractability}
Fix a constant $C$.  Let $\mathcal{U}(C)$ denote the set of non-decreasing functions $G$ where $G(0) = 0$,
$G(1) > 0$, and:

\vspace{1mm}
$\forall k \ \exists N_0,t \leq C \ \forall x,y \in \mathbb{N^+} \ \forall R \in \mathbb{R^+} \ \forall \epsilon:$
\vspace{-4mm}
\begin{equation}
\left(R > N_0, \frac{G(x)}{G(y)} = R, \epsilon > \frac{1}{log^k(Rx)} \right) \Rightarrow \left( \left(\frac{\pi_\epsilon(x)}{y}\right)^2 \geq \frac{R}{log^t(Rx)} \right) \hspace{3mm} and
\end{equation}
\vspace{-2mm}
\begin{equation}
\hspace{-2.049in}\forall k \ \exists r \leq C,N_1 \ \forall x \geq N_1 \ \forall \epsilon : \epsilon > \frac{1}{log^k(x)} \Rightarrow \pi_\epsilon(x) \geq \frac{x}{\log^r(x)}. \\
\end{equation}

\end{defn}

\begin{defn}[Universal Core Structure]\label{defn:ucs}
For a fixed vector $V = (v_1,\ldots,v_n)$, we say a data structure $S$ is a universal core structure with
parameters $\epsilon > 0$, $\delta > 0$, $\alpha > 0$, and a class of functions $\mathcal{G}$,
where $G \in \mathcal{G}$
satisfies $G:\mathbb{R} \rightarrow \mathbb{R}$, if given any $G \in \mathcal{G}$, $S$ outputs a set
$T = \{(x_1,j_1),\ldots,(x_\ell,j_\ell)\}$ such that with probability at least $1-\delta$ we have:
$1)$ For each $1 \leq i \leq \ell$, $(1-\epsilon)G(v_{j_i}) \leq x_i \leq (1+\epsilon)G(v_{j_i})$, and
$2)$ If there exists $i$ such that $v_i$ is $(G,\alpha)$-heavy with respect to $V$, then
$i \in \{j_1,\ldots,j_\ell\}$.
\end{defn}

\begin{defn}[Universal Core Algorithm]\label{defn:uca}
We say an algorithm $\mathcal{A}$ is a universal core algorithm with parameters $\epsilon > 0$, $\delta > 0$,
$\alpha > 0$,
and a class of functions $\mathcal{G}$, where $G \in \mathcal{G}$
satisfies $G:\mathbb{R} \rightarrow \mathbb{R}$, if, given any stream $D(n,N)$ as input,
$\mathcal{A}$ outputs a universal core structure for the vector $M(W)$ with the same parameters
$\epsilon$, $\delta$, $\alpha$, and~$\mathcal{G}$.
\end{defn}

\begin{defn}[Universal Sum Structure]\label{defn:uss}
For a fixed vector $V = (v_1,\ldots,v_n)$, we say a data structure $S$ is a universal sum structure with
parameters $\epsilon > 0$, $\delta > 0$, and a class of functions $\mathcal{G}$, where $G \in \mathcal{G}$
satisfies $G:\mathbb{R} \rightarrow \mathbb{R}$, if given any $G \in \mathcal{G}$, $S$ outputs a value $x$
such that with probability at least $1-\delta$ we have:
$(1-\epsilon)\sum_{i=1}^nG(v_i) \leq x \leq (1+\epsilon)\sum_{i=1}^nG(v_i)$.
\end{defn}

\begin{defn}[Universal Sum Algorithm]\label{defn:usa}
We say an algorithm $\mathcal{A}$ is a universal sum algorithm with parameters $\epsilon > 0$, $\delta > 0$,
and a class of functions $\mathcal{G}$, where $G \in \mathcal{G}$
satisfies $G:\mathbb{R} \rightarrow \mathbb{R}$, if, given any stream $D(n,N)$ as input,
$\mathcal{A}$ outputs a universal sum structure for the vector $M(W)$ with parameters $\epsilon$, $\delta$,
and~$\mathcal{G}$.
\end{defn}

In this paper, our main result is the proof of the following theorem:
\begin{thm}\label{thm:mainusum}
Fix a constant $C$ and let $\mathcal{U}(C)$ be the universally tractable set according to Definition~\ref{defn:utractability}.
There is a universal sum algorithm with parameters
$\epsilon = \Omega(1/\log^k(nN))$ (for any $k \geq 0$), $\delta = 0.3$,
and $\mathcal{G} = \mathcal{U}(C)$.  The algorithm uses polylogarithmic space in $n$ and $N$, and
makes a single pass over the input stream $D(n,N)$.
\end{thm}

\noindent We can reduce the constant failure probability to inverse polynomial via standard methods.

Along the way, we also design a zero-one law (i.e., a test) which, given a function $G$,
determines if it is possible to solve the $G$-Sum problem
using polylogarithmic space in $n$ and $N$ while making one pass over the stream $D$.  If
$G$ passes the test, we give an explicit algorithm which achieves a $(1 \pm \epsilon)$-approximation
except with negligible probability (making only one pass and using polylogarithmic memory).
To formalize our other main result, we define the following class {\sf STREAM-POLYLOG}:
\begin{defn}[{\sf STREAM-POLYLOG}]\label{def: strplg}
We say function $G \in$ {\sf STREAM-POLYLOG} if $\forall k = O(1)$, $\exists t=O(1)$ and
an algorithm $\mathcal{A}$ such that for any universe size $n$, window size $N$, $\epsilon \ge 1/\log^k(nN)$,
and stream $D(n,N)$:
$1)$ $\mathcal{A}$ makes one pass over $D$, $2)$ $\mathcal{A}$ uses $O(\log^{t}(nN))$ space, and $3)$
For any window $W$, $\mathcal{A}$ maintains a $(1 \pm \epsilon)$-approximation of $|G(M(W))|$
except with probability at most $0.3$.
\end{defn}

\noindent Note that the constant error probability can be made to be as small as an inverse polynomial
by standard techniques.  Our other main result is the following theorem:
\begin{thm}\label{thm:zeroone}
Let $G$ be a non-decreasing function such that $G(0) = 0$:
$G\in \mbox{{\sf STREAM-POLYLOG}} \Longleftrightarrow G \in \mbox{\sf Tractable}$.
\end{thm}

\section{Lower Bound for Sliding Windows}
\label{sec:lb}

In this section, we give a space lower bound for any non-tractable function $G$.
We first show a deterministic lower bound for any algorithm that
approximates the $G$-Sum problem.  Our technique is inspired by the lower bound proof in~\cite{DGIM02} for
estimating the number of $1$'s for sliding windows.

\begin{thm}
Let $G$ be a function such that $G \not\in \mathsf{Tractable}$.  Then, any deterministic algorithm that
solves the $G$-Sum problem with relative error $\epsilon' = 1/\log^b(nN)$ (for some constant $b$) must
use space at least $\Omega(\log^a(nN))$, where $a$ is arbitrarily large.
\end{thm}

\begin{proof}
We construct a set of input streams such that, for any pair of data streams in the set, the algorithm must distinguish
between these two inputs at some point as the window slides.  Therefore, the space of the algorithm must be at least
logarithmic in the size of this set.

Since $G \not\in \mathsf{Tractable}$, in Definition~\ref{defn:tractability}, either
Predicate (\ref{eqn:tract1}) or (\ref{eqn:tract2}) does not hold.  If Predicate (\ref{eqn:tract1})
does not hold, then the lower bound from~\cite{BO10} applies and we are done.
Hence, we assume that
Predicate (\ref{eqn:tract2}) does not hold, implying:
$\exists k, \forall r,N_1, \exists x \geq N_1, \epsilon : \epsilon > \frac{1}{log^k(x)} \wedge
\pi_\epsilon(x) < \frac{x}{\log^r(x)}.$
Let $k$ be given, and let $r$ be arbitrarily large.  This negation implies that there are infinitely many
increasing points $x_1,x_2,x_3,\ldots$ and
corresponding values $\epsilon_1,\epsilon_2,\epsilon_3, \ldots$, where $\epsilon_i > \frac{1}{\log^k(x_i)}$
and $\pi_{\epsilon_i}(x_i) < \frac{x_i}{\log^r(x_i)}$.

Surprisingly, we will construct our lower bound with a universe of size $n=1$, namely $U = \{1\}$.  For each $x_i$,
we will construct a set of streams with a fixed, upper bounded window size of $N=x_i$, and argue that the algorithm
must use memory at least $\log^r(x_i)$ (note that, as the $x_i$ are monotonically increasing, our lower bound will apply
for asymptotically large $N$).  We assume without loss of generality that
$G(x_i - \pi_{\epsilon_i}(x_i)) < (1-\epsilon_i)G(x_i)$.  Our constructed streams will be as follows.
For each $N = x_i$, note that our window
consists of elements which have arrived in the past $x_i$ time steps.  For the first $x_i$ time steps, we
construct many streams by choosing $\lfloor\frac{x_i}{\pi_{\epsilon_i}(x_i)}\rfloor$ of these time steps
(each choice defining a different stream).  For each chosen time step, we insert $\pi_{\epsilon_i}(x_i)$ $1$'s into
the stream, and for each time step that is not chosen, we insert zero elements.  For technical reasons, we pad
the last time step $x_i$ in the first window with
$x_i - \pi_{\epsilon_i}(x_i)\lfloor\frac{x_i}{\pi_{\epsilon_i}(x_i)}\rfloor$ $1$'s.  Note that the
number of elements in the first window at time $x_i$ is
$\pi_{\epsilon_i}(x_i)\lfloor\frac{x_i}{\pi_{\epsilon_i}(x_i)}\rfloor +
(x_i - \pi_{\epsilon_i}(x_i)\lfloor\frac{x_i}{\pi_{\epsilon_i}(x_i)}\rfloor) = x_i$.  We insert nothing at
time step $x_i + 1$.  For the remaining time steps $x_i + 2, \ldots, 2x_i - 1$, we simply repeat the first
$x_i - 2$ time steps of the stream (i.e., if time step $t$ was chosen in the first $x_i$ time steps,
$1 \leq t \leq x_i - 2$, then we insert $\pi_{\epsilon_i}(x_i)$ $1$'s at time step $x_i + t + 1$).

Now, we argue that for any such pair of constructed streams $A$, $B$ which are different, any algorithm
with relative error smaller than $\epsilon' = 1/\log^k(nN)$ must distinguish between these two inputs.
To see this, consider the earliest time $d$ when the two streams differ (note that $1 \leq d \leq x_i - 1$).
Let $W_A$ be the window for stream $A$ (and similarly $W_B$ for stream $B$).
Let $c$ be the number of chosen time steps in the first $d$ time steps of stream $A$.
Without loss of generality, we assume time step $d$ was chosen in stream $A$
but not stream $B$.  Hence, the number of chosen time steps in stream $B$ up to time
$d$ is $c-1$.  Consider the windows at time step $x_i + d$.  The number of elements in $W_A$ at this time
is given by $\pi_{\epsilon_i}(x_i)[\lfloor\frac{x_i}{\pi_{\epsilon_i}(x_i)}\rfloor - c + (c-1)] + x_i -
\pi_{\epsilon_i}(x_i)\lfloor\frac{x_i}{\pi_{\epsilon_i}(x_i)}\rfloor = x_i - \pi_{\epsilon_i}(x_i)$.  Moreover,
the number of elements in $W_B$ is given by
$\pi_{\epsilon_i}(x_i)[\lfloor\frac{x_i}{\pi_{\epsilon_i}(x_i)}\rfloor - (c-1) + (c-1)] + x_i -
\pi_{\epsilon_i}(x_i)\lfloor\frac{x_i}{\pi_{\epsilon_i}(x_i)}\rfloor = x_i$.
Hence, the $G$-Sum value at time $x_i + d$ for $W_A$ is $G(x - \pi_{\epsilon_i}(x_i)) < (1-\epsilon_i)G(x_i)$.
As long as the algorithm has relative error $\epsilon ' = 1/\log^k(nN) < \epsilon_i$, streams $A$ and
$B$ must be distinguished.

Thus, the algorithm's memory is lower bounded by the logarithm of the number of constructed streams, of which
there are $\binom{x_i}{\lfloor\frac{x_i}{\pi_{\epsilon_i}(x_i)}\rfloor}$ for each $x_i$.  We have
$\log(\binom{x_i}{\lfloor\frac{x_i}{\pi_{\epsilon_i}(x_i)}\rfloor}) \geq
\lfloor\frac{x_i}{\pi_{\epsilon_i}(x_i)}\rfloor \log (\pi_{\epsilon_i}(x_i)) \geq
\frac{\log^r(x_i)}{2}\log(\pi_{\epsilon_i}(x_i)).$
If $\pi_{\epsilon_i}(x_i) = 1$, we repeat the proof inserting two $1$'s at each time step and the proof goes through.
Observing that $r$ can be made arbitrarily large gives the proof.
\end{proof}

We now have a randomized lower bound by appealing to Yao's minimax principle~\cite{MR95} and building on top of our
deterministic lower bound, similarly to~\cite{DGIM02}.
\begin{thm}\label{thm:rlb}
Let $G$ be a function where $G \not\in \mathsf{Tractable}$.  Then, any randomized algorithm that solves
$G$-Sum with relative error smaller than $\epsilon' = 1/\log^b(nN)$ for some constant $b$ and
succeeds with at least constant probability $1 - \delta$ must use memory $\Omega(\log^a(nN))$, where $a$
is arbitrarily large.
\end{thm}

\section{An Algorithm for Tractable Functions}\label{sec:tract}
In this section, we complete the proof of Theorem~\ref{thm:zeroone} by first approximating heavy elements:
\begin{prob}[$G$-Core]\label{prob:heavy}
We have a stream $D(n,N)$ and parameters~$\epsilon,\delta > 0$.  For each window $W$, with probability at least $1-\delta$,
maintain a set $S = \{g_1',\ldots,g_\ell'\}$ such that
$\ell = O^*(1)$ and there exists a set of indices $\{j_1,\ldots,j_{\ell}\}$ where $(1-\epsilon)G(m_{j_i}) \leq g_i' \leq (1+\epsilon)G(m_{j_i})$
for each $1 \leq i \leq \ell$.  If there is a $(G,1)$-heavy element $m_i$ with respect to $M(W)$, then $i \in \{j_1,\ldots,j_{\ell}\}$.
\end{prob}

We begin solving the above problem with the following lemma (taken from~\cite{BO10}).

\begin{lem}\label{lem:majorseparation}
Let $V = (v_1,\ldots,v_n)$ be a vector with non-negative entries of dimension $n$ and $H$ be a pairwise independent random vector
of dimension $n$ with entries $h_i \in \{0,1\}$ such that $P(h_i = 1) = P(h_i = 0) = \frac{1}{2}$.   Denote
by $H'$ the vector with entries $1 - h_i$.  Let $K > 10^4$ be a constant, and let $X = \langle V,H \rangle$
and $Y = \langle V,H' \rangle$.
If there is an $(F_1,K)$-heavy element $v_i$ with respect to $V$, then:
$P((X > KY) \vee (Y > KX)) = 1.$
If there is no $(F_1,\frac{K}{10^4})$-heavy element with respect to $V$, then:
$P((X > KY) \vee (Y > KX)) \leq \frac{1}{2}.$
\end{lem}

We now prove some lemmas related to how approximating values can affect the function $G$.

\begin{lem}\label{lem:Gapprox}
Let $0 < \epsilon \leq \frac{1}{2}$, and let $x,u,v,y \geq 0$ satisfy $|x - u| \leq 0.1\pi_\epsilon(x)$ and
$v,y < 0.1 \pi_\epsilon(x)$, where $\pi_\epsilon(x) > 1$.
Then $(1 - 4\epsilon)G(u+v+y) \leq G(u) \leq (1+4\epsilon)G(u-v-y)$.
\end{lem}
\begin{proof}
First, we note that $u+v+y \leq x + 0.1 \pi_\epsilon(x) + v + y \leq
x + 0.3 \pi_\epsilon(x) \leq x + \pi_\epsilon(x) - 1$ (recalling $\pi_\epsilon(x) > 1$).  We can
similarly get that $u-v-y \geq x - (\pi_\epsilon(x) - 1)$.
Hence, we get that $(1-\epsilon)G(x) \leq G(x - (\pi_\epsilon(x) - 1)) \leq G(u-v-y) \leq G(u) \leq G(u+v+y)
\leq G(x + (\pi_\epsilon(x) - 1)) \leq (1+\epsilon)G(x)$.

We conclude the proof by noting: $(1+4\epsilon)G(u-v-y) \geq (1+4\epsilon)(1-\epsilon)G(x)
\geq \frac{(1+4\epsilon)(1-\epsilon)}{1+\epsilon}G(u) \geq G(u)$.  Similarly, we get
$(1-4\epsilon)G(u+v+y) \leq (1-4\epsilon)(1+\epsilon)G(x) \leq \frac{(1-4\epsilon)(1+\epsilon)}{1-\epsilon}G(u)
\leq G(u)$.
\end{proof}

\begin{lem}\label{lem:Gapprox2}
Let $x,u,v,y \geq 0$ be such that $|x - u| \leq v + y$, and let $0 < \epsilon < 1$.  If
$(1 - \epsilon)G(u+v+y) \leq G(u) \leq (1+\epsilon)G(u-v-y)$, then
$(1 - \epsilon)G(x) \leq G(u) \leq (1+\epsilon)G(x).$
\end{lem}

\begin{proof}
We have $(1-\epsilon)G(x) \leq (1-\epsilon)G(u+v+y) \leq G(u) \leq (1+\epsilon)G(u-v-y) \leq (1+\epsilon)G(x)$.
\end{proof}

We now give a useful subroutine over sliding windows which we will use in our main algorithm for
approximating heavy elements and prove its correctness
(there is a similar algorithm and proof in~\cite{BO10}, though it must be adapted to the sliding window setting).
\begin{algorithm}
\label{alg:res}
\DontPrintSemicolon
\For{$i = 1$ to $O(\log(nN))$} {
	\For{$j = 1$ to $C = O(1)$} {
		Generate a random hash function $H : [n] \rightarrow \{0,1\}$ with pairwise independent entries.\;
		Let $H' = 1 - H$ (i.e., $h_i' = 1 - h_i$, where $h_i$ is the $i^{th}$ entry of $H$). \;
		Consider the smooth histogram method for approximating $F_2$ on sliding windows~\cite{BO07}. \;
		Let $f_H$ be a $(1 \pm .1)$-approximation of $F_2$ on $D_H$ (with negligible error probability). \;
		Let $f_{H'}$ be a $(1 \pm .1)$-approximation of $F_2$ on $D_{H'}$ (with negligible error probability). \;
		Let $X_{ij} = 10\min\{f_H,f_{H'}\}$.
	}
	Let $Y_i = \frac{X_{i1} + \cdots + X_{iC}}{C}$ (i.e., $Y_i$ is the average of $C$ independent $X_{ij}$'s).\;
}
Output $r = \sqrt{median_i \{ Y_i \}}$ for the current window $W$. \;
\caption{Residual-Approximation($D$)}
\end{algorithm}
\begin{lem}\label{lem:resapprox}
Let $D(n,N)$ be any stream.  Algorithm Residual-Approximation makes a single pass over $D$ and uses
polylogarithmic memory in $n$ and $N$.  Moreover, if the current window $W$ contains an
$(F_2,2)$-heavy element $m_i$ with respect to $M(W)$, then the algorithm maintains and outputs
a value $r$ such that $2\sqrt{F_2^{res}(M(W))} < r < 3\sqrt{F_2^{res}(M(W))}$ (except
with negligible probability).
\end{lem}
\begin{proof}
Assume the current window $W$ has an $(F_2,2)$-heavy element $m_k$ with respect to $M(W)$.  Due to the properties of smooth
histograms from~\cite{BO07}, we know that $.9F_2(M(W_H)) \leq f_H \leq 1.1F_2(M(W_H))$, where $M(W_H)$ is the multiplicity vector
of the current window in substream $D_H$ (similarly for $f_{H'}$).  Hence, the random variable $X_{ij} = 10\min\{f_H,f_{H'}\}$
is a $(1 \pm .1)$-approximation of the random variable $Z = 10\sum_{\ell} {\bf 1}_{H(\ell) \neq H(k)} m_\ell^2$ (here, ${\bf 1}_{H(\ell) \neq H(k)}$
is the indicator random variable which is 1 if $H(\ell) \neq H(k)$ and 0 otherwise).  To see why, suppose that element $k$ is mapped to $1$ by $H$,
so that $k$ belongs to the substream $D_H$.  Then observe that
$$f_H \geq .9F_2(M(W_H)) \geq .9m_k^2 > 1.8\sum_{\ell \neq k}m_\ell^2 \geq 1.1\sum_{\ell}{\bf 1}_{H(\ell) \neq H(k)}m_\ell^2 \geq f_{H'}.$$
Hence, the minimum of $f_H$ and $f_{H'}$ will indeed be a $(1 \pm .1)$-approximation to $\sum_{\ell} {\bf 1}_{H(\ell) \neq H(k)} m_\ell^2$,
since this is the second moment of the vector $M(W_{H'})$ (the case is symmetric if element $k$ is mapped to $0$ by $H$).

Now, since $H$ is pairwise independent, we have that $E(Z) = 5F_2^{res}(M(W))$.  In particular, since we always have $0 \leq Z \leq 10F_2^{res}(M(W))$,
we can bound the variance by $Var(Z) \leq E(Z^2) \leq 100 (F_2^{res}(M(W)))^2$.  If we denote by $A$ the random variable which is the average
of $C$ independent $Z$'s, then we have $Var(A) = \frac{1}{C}Var(Z) \leq \frac{100}{C}(F_2^{res}(M(W)))^2$.  Hence, if we choose $C$ to be sufficiently
large, then by Chebyshev's inequality we have:
$$ P(|A - 5F_2^{res}(M(W))| \geq 0.1F_2^{res}(M(W))) \leq \frac{100Var(A)}{(F_2^{res}(M(W)))^2} \leq \frac{10^4}{C} \leq 0.1$$
(for instance, $C = 10^5$ is sufficient).

Now, if we take the median $T$ of $O(\log (nN))$ independent $A$'s, then by Chernoff bound this would make the probability negligible.
That is, we have $4.9F_2^{res}(M(W)) \leq T \leq 5.1F_2^{res}(M(W))$ except with negligible probability.  We can
repeat these arguments and consider the median of $O(\log(nN))$ averages (i.e., the $Y_i$'s) of $O(1)$ independent $X_{ij}$'s.  Since
there are only $O(\log(nN))$ $X_{ij}$'s total (with each one being a $(1 \pm .1)$-approximation to its corresponding random variable $Z$,
except with negligible probability), then by the union bound all of the $X_{ij}$'s will be $(1 \pm .1)$-approximations except with negligible probability
(since the sum of polylogarithmically many negligible probabilities is still negligible).  Therefore, the median of averages would give a
$(1 \pm .1)$-approximation to $T$.  Taking the square root guarantees that $2\sqrt{F_2^{res}(M(W))} < r < 3\sqrt{F_2^{res}(M(W))}$ (except with negligible
probability).

Note that the subroutine for computing an approximation to $F_2$ on sliding windows using smooth histograms can be done in one pass and
in polylogarithmic space (even if we demand a $(1 \pm .1)$-approximation and a negligible probability of failure).
\end{proof}

Now, we claim that Algorithm Compute-Hybrid-Major solves the following problem:
\begin{prob}[Hybrid-Major$(D,\epsilon)$]
Given a stream $D$ and $\epsilon > 0$, maintain a value $r \geq 0$ for each window $W$ such
that:
$1)$ If $r\neq 0$, then $r$ is a $(1\pm 4\epsilon)$-approximation of $G(m_j)$ for some $m_j$, and
$2)$ If the current window $W$ has an element $m_i$ such that
$\pi_{\epsilon}(m_i) \ge 20^5\sqrt{F^{res}_2(M(W))}$, then $r$ is a $(1\pm 4\epsilon)$-approximation of $G(m_i)$.
\end{prob}

\begin{algorithm}\label{alg:heavyapprox}
\DontPrintSemicolon
Let $a$ be a $(1 \pm \epsilon')$-approximation of $L_2$ for window $W$ using the smooth histogram
method~\cite{BO07} (with negligible probability of error), where $\epsilon' = \frac{1}{\log^{\Omega(1)}(N)}$. \;

Repeat $O(\log(nN))$ times, independently and in parallel: \;
\Indp
Generate a uniform pairwise independent vector $H \in \{0,1\}^n$. \;
Maintain and denote by $X'$ a $(1 \pm .2)$-approximation of the second moment for the window $W_H$ using
a smooth histogram~\cite{BO07} (with negligible probability of error). \;
Similarly define $Y'$ for the window $W_{\one-H}$. \;
If $X' < {(20)^4} Y'$ and $Y' < {(20)^4} X'$, then output $0$ and terminate the algorithm. \;
\Indm
In parallel, apply Residual-Approximation($D$) to maintain the residual second moment approximation,
let $b$ denote the output of the algorithm. \;
If $(1-4\epsilon)G(a + b + 2 \epsilon' a) > G(a)$ or $G(a) > (1+4\epsilon)G(a - b - 2 \epsilon' a)$, then output $0$. \;
Otherwise, output $G(a)$. \;
\caption{Compute-Hybrid-Major$(D, \epsilon)$}
\end{algorithm}

Before delving into the proof, we first give the following useful lemma.

\begin{lem}\label{lem:heavyapprox}
Suppose the current window $W$ contains an $(F_2,1)$-heavy element $m_i$.
Moreover, let $a$ be a $(1 \pm \epsilon')$-approximation of the $L_2$ norm of the current window $W$, where $\epsilon' < 1$.
Then $-\epsilon' m_i \leq a - m_i \leq (1 + \epsilon')\sqrt{F^{res}_2(M(W))} + \epsilon' m_i \leq 2\sqrt{F^{res}_2(M(W))} + \epsilon' m_i$.
\end{lem}
\begin{proof}
Since $a$ is a $(1 \pm \epsilon')$-approximation of the $L_2$ norm of $M(W)$, we have
$(1-\epsilon')\sqrt{\sum_{k=1}^nm_k^2} \leq a \leq (1+\epsilon')\sqrt{\sum_{k=1}^nm_k^2}$.
Hence, we have that
$$
a - m_i \leq (1+\epsilon')\sqrt{\sum_{k=1}^n m_k^2} - m_i \leq (1+\epsilon')m_i + (1+\epsilon')\sqrt{\sum_{j \neq i}m_j^2} - m_i \leq \epsilon' m_i + (1+\epsilon')\sqrt{F_2^{res}(M(W))}.
$$
Moreover, we have $m_i - a \leq m_i - (1-\epsilon')\sqrt{\sum_{k=1}^n m_k^2} \leq m_i - (1-\epsilon')m_i$,
which gives the lemma.
\end{proof}

\begin{lem}\label{lem:hybridmajor}
Compute-Hybrid-Major solves Hybrid-Major$(D, \epsilon)$ with negligible probability of error.
\end{lem}
\begin{proof}

First, we show that if there is no $(F_2,2)$-heavy entry in the current window $W$, then the output
is $0$ except with negligible probability.  Consider a single iteration of the main loop of the algorithm.
Let $M'$ be the vector with entries $m_i^2$ and denote
$X = \langle M', H \rangle, Y = |M'| - \langle M', H \rangle$.
Since we have an $F_2$ approximation over sliding windows, except with negligible probability, $X'$ and $Y'$ are $(1 \pm .2)$-approximations
of $X$ and $Y$, respectively.  Hence, $\frac{4}{5}X \leq X' \leq \frac{5}{4}X$ and $\frac{4}{5}Y \leq Y' \leq \frac{5}{4}Y$.
By Lemma \ref{lem:majorseparation}, except with probability at most $0.5+o(1)$:
$X' \le \frac{5}{4} X \le \frac{5}{2} (10)^4 Y < (20)^4  Y'$ and $Y' < (20)^4 X'$.
Thus, the algorithm outputs $0$ except with negligible probability.

Assume that there is an $(F_2,2)$-heavy entry $m_i$.
Then, applying Lemma~\ref{lem:heavyapprox} with some $0 < \epsilon' < 1$ to be set later, we know
$|m_i - a| \le 2\sqrt{F_2(M(W))} + \epsilon' m_i$ and $a \geq (1-\epsilon')m_i$
(except with negligible probability).
By Lemma~\ref{lem:resapprox},
it follows that $2\sqrt{F^{res}_2(M(W))} < b < 3\sqrt{F^{res}_2(M(W))}$ except with negligible probability.
Hence, we have $|m_i - a| \leq b + \epsilon' m_i \leq b + 2 \epsilon' a$, since
$2 \epsilon' a \geq 2 \epsilon' (1- \epsilon')m_i \geq \epsilon' m_i$ (assuming $\epsilon' \leq \frac{1}{2}$).
Now, observe that if the algorithm outputs $G(a)$, then it must be that
$(1-4\epsilon)G(a+b+2 \epsilon' a) \leq G(a) \leq (1+4\epsilon)G(a-b-2 \epsilon' a)$.
Thus, by applying Lemma~\ref{lem:Gapprox2} with parameters $x=m_i, u = a, v = b$, and $y=2 \epsilon' a$,
it follows that if the algorithm outputs $G(a)$, then $G(a)$ is a $(1\pm 4\epsilon)$-approximation of $G(m_i)$.
Thus, the first condition of Hybrid-Major follows.

Finally, assume $\pi_{\epsilon}(m_i) \ge (20)^5\sqrt{F^{res}_2(M(W))}$. By definition,
$m_i \ge \pi_{\epsilon}(m_i)$ and so $m_i$ is $(F_2,20^{10})$-heavy with respect to $M(W)$.
By Lemma~\ref{lem:majorseparation}, we have (except with negligible probability):
$X' > 20^4 Y' \text{ or } Y' > 20^4 X'$.
Hence, except with negligible probability, the algorithm will not terminate before the last line.
Let $N_1$ be the constant given by the definition of
tractability in Definition~\ref{defn:tractability}.  We assume $m_i \geq N_1$ (otherwise
the number of elements in the window is constant).  Also, let $r$
be given by Definition~\ref{defn:tractability}.
By applying Lemma~\ref{lem:heavyapprox} with $\epsilon' = \frac{1}{\log^{r+1}(N)}$, we have
$
|m_i - a| \le 2\sqrt{F^{res}_2(M(W))} +  \epsilon' m_i \leq 0.01\pi_{\epsilon}(m_i) + \frac{1}{\log N}\frac{m_i}{\log^r(m_i)} \leq
.01 \pi_{\epsilon}(m_i) + \frac{\pi_{\epsilon}(m_i)}{\log N} \leq .02 \pi_{\epsilon}(m_i)$ for sufficiently large $N$ (since $G$
is tractable) and
$b \le 3F^{res}_2(M(W)) < 0.01\pi_{\epsilon}(m_i)$.
Since $b \leq .1 \pi_{\epsilon}(m_i)$
and $2 \epsilon' a \leq 2 \epsilon' (m_i + b + \epsilon' m_i) \leq 2 \cdot (.03 \pi_{\epsilon}(m_i)) \leq .1 \pi_{\epsilon}(m_i)$,
then by Lemmas~\ref{lem:Gapprox} and~\ref{lem:Gapprox2} (which we apply with the same parameters,
$x=m_i, u = a, v = b$, and $y=2 \epsilon' a$), the algorithm outputs $G(a)$ which is a $(1 \pm 4\epsilon)$-approximation
of $G(m_i)$.  Thus, the second condition of Hybrid-Major follows.
\end{proof}

\noindent
The next two lemmas are from~\cite{BO10} (the proof of Lemma~\ref{lem:G-major} uses Predicate (\ref{eqn:tract1})
from Definition~\ref{defn:tractability}).
\begin{lem}\label{lem:tractcube}
If $G$ is tractable, then
$\exists N_1 \forall N>N_1 \in \mathbb{N^+}: G(N) \le N^3$.
\end{lem}
\begin{lem}\label{lem:G-major}
Let $G$ be a non-decreasing tractable function. Then for any $k = O(1)$, there
exists $t = O(1)$ such that for any $n,N$ and for any $\epsilon > \log^{-k}(nN)$
the following holds. Let $D(n,N)$ be a stream and $W$ be the current window.
If there is a $(G,1)$-heavy element $m_i$ with respect to $M(W)$, then there is
a set $S \subseteq [n]$ such that $|S| = O(\log(N))$ and:
$
\pi^2_{\epsilon}(m_i) = \Omega \left( \log^{-t}(nN)\sum_{j\notin S\cup\{i\}} m_j^2 \right).
$
\end{lem}

We now give the algorithm Compute-$G$-Core, which solves the $G$-Core problem
(i.e., Problem~\ref{prob:heavy}), and prove its correctness.
A similar algorithm appears in~\cite{BO10}, we
repeat it here for completeness, and to help design and understand our main result on universality.

\begin{algorithm}\label{alg:heavy}
\DontPrintSemicolon
Generate a pairwise independent hash function $H: [n] \mapsto \tau$, where $\tau = O^*\left({1\over p}\right)$. \;
$\forall k\in [\tau]$, compute in parallel $c_k$ = Compute-Hybrid-Major$(D_{H_k}, \frac{\epsilon}{4})$,
where $H_k(i) = \one_{H(i) = k}$. \;
Output $S = \{c_k : c_k > 0\}$. \;
\caption{Compute-$G$-Core$(D, \epsilon, p)$}
\end{algorithm}

\begin{thm}\label{thm:gcore}
Algorithm Compute-$G$-Core solves $G$-Core$(D, \epsilon, p)$, except with probability asymptotically equal to $p$.
Compute-$G$-Core uses $O^*(1)$ memory bits if $p = \Omega(1/\log^{u}(nN))$ and $\epsilon = \Omega(1/\log^{k}(nN))$
for some $u,k \geq 0$.
\end{thm}

\begin{proof}
Let $W$ denote the current window.  First, except with negligible probability, every positive $c_i$ is a
$(1 \pm 4 \cdot \frac{\epsilon}{4})$-approximation of some distinct entry $G(m_{j})$, which implies that
$c_i$ is a $(1 \pm \epsilon)$-approximation of $G(m_j)$.  Second, assume that there exists a $(G,1)$-heavy
entry $m_i$ with respect to $M(W)$. Denote $X = \sum_{j \neq i} m_j^2 \one_{H(j) = H(i)}$. By pairwise independence of $H$, we have
$E(X) = {1\over \tau}(F_2(M) - m_i^2)$.  By Lemma \ref{lem:G-major}, there exists a set $S$ and $t \geq 0$ such that
$|S| = O(\log N)$ and:
\begin{equation}\label{eq:G-major-eqn}
\pi^2_{\epsilon}(m_i) = \Omega\left(\frac{\sum_{j\notin S\cup\{i\}} m_j^2}{\log^{t}(nN)}\right).
\end{equation}
Let $\mathcal{L}$ be the event that $\pi^2_{\epsilon}(m_i) > 20^{10} X$, and let $\mathcal{B}$ be the event that $\forall j \in S : H(j) \neq H(i)$.
By Markov's inequality, by pairwise independence of $H$, and by $(\ref{eq:G-major-eqn})$, there exists $\tau = O^*\left({1\over p}\right)$ such that:
$$
P(\neg\mathcal{L}) = P(\neg\mathcal{L} \ | \ \mathcal{B} ) \cdot P(\mathcal{B}) + P(\neg\mathcal{L} \ | \ \neg\mathcal{B}) \cdot P(\neg \mathcal{B}) \leq
\frac{E(X \ | \ \mathcal{B}) 20^{10}}{\pi^2_\epsilon(m_i)} \cdot 1 + 1 \cdot \frac{O(\log N)}{\tau} \leq O^*\left(\frac{1}{\tau}\right) = p.
$$
If $\mathcal{L}$ occurs (which happens with probability at least $1-p$), then $c_{H(i)}$ is a $(1 \pm \epsilon)$-approximation of $G(m_i)$
except with negligible probability (by Lemma~\ref{lem:hybridmajor}).  Thus, the final probability of error is approximately equal to $p$.

It is not too hard to see that Algorithm Compute-$G$-Core uses polylogarithmic memory.  The subroutine depth is constant, and there are only
polylogarithmically many subroutine calls at each level.  At the lowest level, we only do direct computations on the stream that require polylogarithmic
space or a smooth histogram computation for $F_2$ or $L_2$, which also requires polylogarithmic space.
We get that for any constant $k$, there exists a constant $t$ such that we can solve $G$-Core$(D, \epsilon,p)$ (except with probability $p$) using
$O(\log^t(nN))$ space, where $\epsilon\ge \log^{-k}(nN)$.
\end{proof}

In Appendix~\ref{apx:sum}, we show how to reduce the $G$-Sum problem to the $G$-Core problem.  In
particular, we prove the following theorem.  The algorithm and proof of correctness follow from~\cite{BO13}.
We restate the algorithm and results using our notation for completeness.
\begin{thm}\label{thm:gsum from gcore}
If there is an algorithm that solves $G$-Core using memory $O^*(1)$ and makes one pass over
$D$ except with probability $O(\log^{-u}(nN))$ for some $u>0$, then there is an algorithm that
solves $G$-Sum using memory $O^*(1)$ and makes one pass over $D$ except with probability at most $0.3$.
\end{thm}

\noindent Note that we can reduce the failure probability from constant to inverse polynomial using standard techniques.
Combining this with Theorem~\ref{thm:gcore} and Theorem~\ref{thm:rlb}, we have Theorem~\ref{thm:zeroone}.

\section{Universality}\label{sec:universality}

In this section, we show the main result of this paper, Theorem~\ref{thm:mainusum},
by designing a universal sum algorithm.  We first construct a universal core algorithm, which we call $UCA$.
That is, given a data stream,
the algorithm produces a universal core structure with respect to the frequency vector $(m_1,\ldots,m_n)$
defined by the current window $W$ without knowing the function $G$ to be approximated in advance.  Let
$C$ be a constant and let $\mathcal{U}(C)$ be the set according to Definition~\ref{defn:utractability}.
The structure will have the guarantee that, when queried with any function $G$ taken from $\mathcal{U}(C)$ (after
processing the stream), it outputs the set $T$ according to Definition~\ref{defn:ucs}.

\paragraph{Universal Core Algorithm ($UCA$):} The algorithm will construct a universal core structure
$S$ and our techniques will build on the results from Section~\ref{sec:tract}.
Algorithm Residual-Approximation from Section~\ref{sec:tract} does not depend on the function $G$,
so it clearly carries over to our universal setting.

Algorithm Compute-Hybrid-Major does depend on $G$, so we have to modify it accordingly.
We do not rewrite the whole algorithm, as there are only a few modifications.  In Step 1,
we set $\epsilon' = \frac{1}{\log^{C+1}(N)}$.  We get rid of Steps 8 and 9, and instead create a new
Step 8 where we find the index $j$ of an $(F_2,2)$-heavy element $m_j$, if it exists (finding
such an index can be done using standard methods, the details of which we omit
for brevity).  We also create a new Step 9 where we output the triple $(a,b,j)$ (assuming none of the parallel
copies from Step 2 outputs 0).

We also need to modify Algorithm Compute-$G$-Core.  In particular, the value
of $\tau$ in Step 1 should depend on $C$, and we set it to be $\frac{\log^{C+2}(nN)}{p}$.  Moreover,
we remove Step 3 from the algorithm and store $c_k$ for each $k \in [\tau]$ as part of our
data structure $S$ (recall that $c_k$ is either 0 or a triple $(a_k,b_k,j_k)$, where $a_k,b_k$ are the values
computed in the $k^{th}$ parallel instance of the subroutine Compute-Hybrid-Major and $j_k$ is the index
of the corresponding $(F_2,2)$-heavy element).

\paragraph{Querying the Structure:} Given a function $G \in \mathcal{U}(C)$ as a query to our universal
core structure, we now explain how to produce the set $T$ according to Definition~\ref{defn:ucs}.

For each stored $c_k$ in the data structure $S$ ($k \in [\tau]$), if $c_k = 0$, then
we do not include it in our output set $T$.  Otherwise, if $c_k$ is a triple $(a_k,b_k,j_k)$,
then we include the pair $(G(a_k),j_k)$ in our set $T$ as long as
$(1-4\epsilon)G(a_k + b_k + 2 \epsilon' a_k) \leq G(a_k) \leq (1+4\epsilon)G(a_k - b_k - 2 \epsilon' a_k)$
(recall $\epsilon' = \frac{1}{log^{C+1}(N)}$).

\begin{thm}\label{thm:ucore}
Fix a parameter $C$ and let $\mathcal{U}(C)$ be the set of tractable functions corresponding to the
definition of universal tractability.  Then $UCA$ is a universal core algorithm with parameters
$\epsilon = \Omega(1/\log^k(nN))$ (for any $k \geq 0$), $\delta = \Omega(1/\log^{u}(nN))$
(for any $u \geq 0$), $\alpha = 1$, and $\mathcal{G} = \mathcal{U}(C)$.
\end{thm}

\begin{proof}
The correctness of $UCA$ essentially follows from the proofs of the results in Section~\ref{sec:tract}.
In particular, Lemma~\ref{lem:resapprox} still holds since Algorithm Residual-Approximation is unchanged.

Lemma~\ref{lem:hybridmajor} still mostly holds without much modification.  Using the same notation
as in the original proof, if there
is no $(F_2,2)$-heavy element, then the proof of Lemma~\ref{lem:hybridmajor} can still be applied
and the modified version of Compute-Hybrid-Major will output $0$
(except with negligible probability).  In such a case, the universal core structure will store the value $0$.
If there is an $(F_2,2)$-heavy element $m_{i_k}$ and the structure stores $(a_k,b_k,i_k)$, then again the
proof applies.  The reason is that, when
querying the universal core structure with a function $G$, we check if
$(1-4\epsilon)G(a_k + b_k + 2 \epsilon' a_k) \leq G(a_k) \leq (1+4\epsilon)G(a_k - b_k - 2 \epsilon' a_k)$,
in which case the proof argues that $G(a_k)$ is a $(1 \pm 4\epsilon)$-approximation of $G(m_{i_k})$.  In
the case that $\pi_{\epsilon}(m_{i_k}) \ge (20)^5\sqrt{F^{res}_2(M(W))}$, the proof still goes through
since we apply Lemma~\ref{lem:heavyapprox} with $\epsilon' = \frac{1}{\log^{C+1}(N)}$, and we have
$|m_{i_k} - a_k| \le 2\sqrt{F^{res}_2(M(W))} +  \epsilon' m_{i_k} \leq 0.01\pi_{\epsilon}(m_{i_k}) +
\frac{1}{\log N}\frac{m_{i_k}}{\log^C(m_{i_k})} \leq 0.01\pi_{\epsilon}(m_{i_k}) +
\frac{1}{\log N}\frac{m_{i_k}}{\log^r(m_{i_k})} \leq .01 \pi_{\epsilon}(m_{i_k}) + \frac{\pi_{\epsilon}(m_{i_k})}{\log N}
\leq .02 \pi_{\epsilon}(m_{i_k})$ (here, similarly to Lemma~\ref{lem:hybridmajor}, $r$ is the constant given by
the definition of universal tractability for $\mathcal{U}(C)$, and hence $r \leq C$).

Finally, we must argue the correctness of Theorem~\ref{thm:gcore}.  Using some notation taken
from the proof, consider an output $c_k = (a_k,b_k,i_k)$ (if $c_k = 0$, the data structure does not output it
to the set $T$) and observe that $G(a_k)$ for any $a_k$ satisfying
$(1-4\epsilon)G(a_k + b_k + 2 \epsilon' a_k) \leq G(a_k) \leq (1+4\epsilon)G(a_k - b_k - 2 \epsilon' a_k)$
is a $(1 \pm 4 \cdot \frac{\epsilon}{4})$-approximation of a distinct entry $G(m_{i_k})$.  Moreover,
if there is a $(G,1)$-heavy element $m_{i_k}$, then we again have
$\pi^2_{\epsilon}(m_{i_k}) = \Omega \left( \log^{-(t+1)}(nN)\sum_{j\notin S\cup\{i_k\}} m_j^2 \right)$.  In
fact, delving into the proof of Lemma~\ref{lem:G-major} (found in~\cite{BO10}), we see that the
specific value of $t$ depends on $G$, and is given by the definition of universal tractability for $\mathcal{U}(C)$.
Since $t \leq C$ and we choose $\tau = \frac{\log^{C+2}(nN)}{p}$, we get the probability of the bad
event $\neg\mathcal{L}$ (using the same notation from Theorem~\ref{thm:gcore}) is bounded by:
$$ \frac{E(X \ | \ \mathcal{B}) 20^{10}}{\pi^2_\epsilon(m_{i_k})} + \frac{O(\log N)}{\tau}
= \frac{20^{10}\log^{t+1}(nN)\sum_{j \notin S\cup\{i_k\}} m_j^2}{\tau \sum_{j \notin S\cup\{i_k\}} m_j^2} +
\frac{O(\log N)}{\tau} \leq p.$$
The rest of the proof goes through in the same way, and hence this gives the theorem.
\end{proof}

We now argue how to use our universal core algorithm $UCA$ as a subroutine to give the main result
of the paper.  The proof of the theorem below can be found in Appendix~\ref{apx:sumcore}, and the argument
follows a similar one found in~\cite{BO13} (we reproduce it in the appendix for completeness).

\begin{thm}\label{thm:usum}
Fix a parameter $C$ and let $\mathcal{U}(C)$ be the set of tractable functions corresponding to the
definition of universal tractability.  Suppose there is a universal core algorithm that makes a single pass over
$D$, uses polylogarithmic memory in $n$ and $N$, and has parameters
$\epsilon = \Omega(1/\log^k(nN))$ (for any $k \geq 0$), $\delta = \Omega(1/\log^{u}(nN))$
(for any $u \geq 0$), $\alpha = 1$, and $\mathcal{G} = \mathcal{U}(C)$.  Then
there is a universal sum algorithm with parameters
$\epsilon = \Omega(1/\log^k(nN))$ (for $k \geq 0$), $\delta = 0.3$,
and $\mathcal{G} = \mathcal{U}(C)$.  The algorithm uses polylogarithmic space in $n$ and $N$ and makes a single
pass over $D$.
\end{thm}

\noindent We can reduce the failure probability of the universal sum algorithm to
inverse polynomial via standard methods.  Our main result, Theorem~\ref{thm:mainusum},
follows from Theorem~\ref{thm:ucore} and Theorem~\ref{thm:usum}.

\bibliographystyle{plain}
\bibliography{ZeroOne}

\appendix

\section{$G$-Sum from $G$-Core}\label{apx:sum}
In this section, we prove Theorem~\ref{thm:gsum from gcore}.
As mentioned, the algorithm and correctness follow from~\cite{BO13}.  We include the algorithm
here for completeness, and rephrase it and the results using notation from our paper.

Let $G$ be a tractable function according to Definition~\ref{defn:tractability}, and let $D(n,N)$ be a stream given as input.
We show how to construct an algorithm that solves the $G$-Sum problem by using our algorithm for $G$-Core as
a subroutine.  In particular, let Compute-$G$-Core be our algorithm from Section~\ref{sec:tract} that solves
the $G$-Core problem.  Note that for the output set 
$S = \{g_1',\ldots,g_\ell'\}$ maintained by Compute-$G$-Core, using standard techniques one can easily obtain the explicit set
of indices $\{j_1,\ldots,j_{\ell}\}$ such that $(1-\epsilon)G(m_{j_i}) \leq g_i' \leq (1+\epsilon)G(m_{j_i})$
for each $1 \leq i \leq \ell$.  Hence, we assume that Compute-$G$-Core outputs a set of pairs of the form
$\{(g_1',j_1),\ldots,(g_\ell',j_\ell) \}$.

In the language of~\cite{BO13}, Compute-$G$-Core produces a $(1,\epsilon)$-cover with respect to the
vector $G(M(W)) = (G(m_1),\ldots,G(m_n))$ with probability at least $1-\delta$, where $\epsilon = \Omega(1/\log^k(nN))$
(for any $k \geq 0$) and $\delta = \Omega(1/\log^{u}(nN))$ (for any $u \geq 0$).  Given the tractable function
$G$, our algorithm for $G$-Sum is as follows:

\bigskip

\begin{algorithm}[H]\label{alg:gsum}
\DontPrintSemicolon
Generate $\phi=O(\log(n))$ pairwise independent, uniform zero-one vectors $H_1, \dots, H_\phi: [n] \rightarrow \{0,1\}$,
and let $h_i^k = H_k(i)$.
Let $D_k$ be the substream defined by $D_{H_1H_2\dots H_k}$, and let $G(M(W_k))$ denote
$(G(m_1),\ldots,G(m_n))$ for the substream $D_k$
and window $W$ (where $k \in [\phi]$). \;

Maintain, in parallel, the cores $Q_k = $ Compute-$G$-Core$(D_k, {\phi^3\over \epsilon^2}, \epsilon, {1\over \phi})$ for each $k \in [\phi]$. \;

If $F_0(G(M(W_\phi))) > 10^{10}$, then output $0$. \;

Otherwise, precisely compute $Y_\phi = |G(M(W_\phi))|$. \;

For each $k = \phi-1,\dots, 0$, compute $Y_k = 2Y_{k+1} - \sum_{(g_i', j_i) \in Q_k} (1-2h_{j_i}^k)g_i'$. \;
Output $Y_0$. \;
\caption{$G$-Sum$(D, \epsilon)$}
\end{algorithm}

Note that, in our paper, Compute-$G$-Core$(D,\epsilon,\delta)$ only takes three parameters (the stream $D$, the
error bound $\epsilon$, and the failure probability $\delta$), while the algorithm from~\cite{BO13} assumes
four parameters of the form Compute-$G$-Core$(D,\alpha,\epsilon,\delta)$.  Here, $D$, $\epsilon$, and $\delta$
have the same meaning as in our paper.  The parameter $\alpha$ controls how heavy an element needs to be
(according to the function $G$) in order to necessarily be in the output set of Compute-$G$-Core.  That is,
in the set $T = \{(x_1,j_1),\ldots,(x_\ell,j_\ell)\}$ output by Compute-$G$-Core, if there is an $i$ such that
$m_i$ is $(G,\alpha)$-heavy with respect to $M(W)$, then $i \in \{j_1,\ldots,j_\ell \}$.  In
Section~\ref{sec:tract}, we solve the problem for $\alpha = 1$, but Algorithm $G$-Sum needs the problem
solved for $\alpha = \frac{\phi^3}{\epsilon^2}$. However, using standard techniques, we can reduce the problem of
solving $G$-Core for $\alpha = \frac{\phi^3}{\epsilon^2}$ to the same problem for $\alpha = 1$.

\begin{thm}

For any tractable function $G$, Algorithm $G$-Sum computes a $(1\pm \epsilon)$-approximation of $|G(M(W))|$
except with probability at most $0.3$, where $\epsilon = \Omega(1/log^k(nN))$ for any $k \geq 0$.
The algorithm uses memory that is polylogarithmic in $n$ and $N$.
\end{thm}

\begin{proof}
The proof of this theorem follows directly from Theorem 1 in~\cite{BO13}.
\end{proof}

\noindent Note that we can turn the constant failure probability into an inverse polynomial failure probability
using standard techniques.

\section{Universal Sum from Universal Core}\label{apx:sumcore}

In this section, we prove Theorem~\ref{thm:usum}. The algorithm and proof are similar to that of
Appendix~\ref{apx:sum}, except that we need to carry out the argument within our universal framework.
As mentioned, the algorithm and correctness follow from~\cite{BO13}.  We do not rewrite the whole algorithm,
but instead describe the necessary modifications that need to be made from Appendix~\ref{apx:sum}.

Let $D(n,N)$ be a stream given as input to our universal sum algorithm.
Let $UCA$ be our universal core algorithm from Theorem~\ref{thm:ucore}, Section~\ref{sec:universality},
the parameters of which are specified in our universal sum algorithm description.

\paragraph{Universal Sum Algorithm:} We describe the modifications that need to be made to Algorithm
$G$-Sum from Appendix~\ref{apx:sum}.

In Step 2, instead we need to maintain and store the
output $Q_k = UCA$ with parameters $\alpha = \frac{\phi^3}{\epsilon^2}$, $\epsilon$ (i.e., the one
given as input to our universal sum algorithm), $\delta = \frac{1}{\phi}$, and
$\mathcal{G} = \mathcal{U}(C)$ for each $k \in [\phi]$ (in the $k^{th}$ parallel iteration, $UCA$ is
given the stream $D_k$ as input).  As in Appendix~\ref{apx:sum}, we construct a universal core
structure for $\alpha = 1$, but we can reduce the problem of $\alpha = \frac{\phi^3}{\epsilon^2}$ to $\alpha = 1$.
Note that $Q_k$ is of the form $\{(a_1,b_1,j_1),\ldots,(a_\ell,b_\ell,j_\ell) \}$ ($Q_k$ may have $0$'s as well,
which we simply ignore).  For each such triple $(a_i,b_i,j_i)$, we also store the value of $h^k_{j_i} = H_k(j_i)$.

In Step 3, instead we check if $F_0(M(W_\phi)) \leq 10^{10}$, and if so we store $M(W_\phi)$ (recall $M(W_\phi)$
denotes the frequency vector $(m_1,\ldots,m_n)$ for the substream $D_\phi$ induced by $W$).  We remove Steps
4, 5, and 6.

\paragraph{Querying the Structure:} Now, given a function $G \in \mathcal{U}(C)$, we explain how to query
the universal sum structure output by our universal sum algorithm to approximate $|G(M(W))|$.  In particular, for each $k$ we first
query the universal core structure output by $UCA$ to get a set $Q'_k = \{(x_1,j_1),\ldots,(x_{\ell'},j_{\ell'})\}$.
Then, we compute $Y_\phi = |G(M(W_\phi))|$ and, for each $k = \phi - 1,\ldots,0$, we recursively compute
$Y_k$ according to:

$$Y_k = 2Y_{k+1} - \sum_{(x_i, j_i) \in Q_k} (1-2h_{j_i}^k)x_i.$$

\noindent Once each $Y_k$ has been computed for $0 \leq k \leq \phi$, we output $Y_0$.

\begin{thm}
Fix a parameter $C$ and let $\mathcal{U}(C)$ be the set of tractable functions corresponding to the
definition of universal tractability.  There is a universal sum algorithm with parameters
$\epsilon = \Omega(1/\log^k(nN))$ (for $k \geq 0$), $\delta = 0.3$,
and $\mathcal{G} = \mathcal{U}(C)$.  The algorithm uses polylogarithmic space in $n$ and $N$ and makes a single
pass over $D$.  When querying the universal sum structure (output by the universal sum algorithm)
with a function $G \in \mathcal{U}(C)$, it outputs a $(1 \pm \epsilon)$-approximation of $|G(M(W))|$ except
with probability at most $0.3$.
\end{thm}

\begin{proof}
The proof of this theorem follows directly from Theorem 1 in~\cite{BO13}.
\end{proof}

\end{document}